\documentclass[a4paper]{article}

\usepackage[english]{babel}
\usepackage{fixmetodonotes}

\usepackage[utf8]{inputenc}
\usepackage{textcomp}
\usepackage{eurosym}
\DeclareUnicodeCharacter{20AC}{\euro}

\makeatletter
\def\blfootnote{\xdef\@thefnmark{}\@footnotetext}
\makeatother

\usepackage{graphicx}
\usepackage{placeins}
\usepackage{subcaption}

\usepackage[shortlabels]{enumitem}

\usepackage[hyphens]{url}
\usepackage{hyperref}
\hypersetup{
    colorlinks,
    allcolors=black,
}

\newcommand{\aref}[1]{\hyperref[#1]{Appendix~\ref{#1}}}

\usepackage{amsmath,amsthm}
\newtheorem{theorem}{Theorem}
\newtheorem{corollary}{Corollary}
\theoremstyle{definition}
\newtheorem{definition}{Definition}

\newtheorem{innerquotedef}{Definition}
\newenvironment{quotedef}[1]
  {\renewcommand\theinnerquotedef{#1}\innerquotedef}
  {\endinnerquotedef}

\DeclareMathOperator{\img}{Img}
\DeclareMathOperator{\dom}{Dom}
\newcommand{\abs}[1]{\left\lVert#1\right\rVert}
\newcommand{\autotag}[1]{\stepcounter{equation}\tag{\theequation}\label{#1}}
\newcommand{\comp}{\mathbin{\Vert}}
\newcommand{\given}{\mathbin{\vert}}
\newcommand{\restr}[2]{{
  \left.\kern-\nulldelimiterspace 
  #1 
  \right|_{#2} 
}}

\usepackage{forest}
\forestset{
    sn/.style={for tree={parent anchor=south, child anchor=north}}
}
\usepackage[pdf, singlefile]{graphviz}

\usepackage{doi}
\usepackage[authoryear,round]{natbib}

\let\citet\Citet
\let\citep\Citep
\let\citeauthor\Citeauthor

\let\cite\CITE

\begin{document}

\def\sectionautorefname{Section}
\def\subsectionautorefname{Section}
\def\subsubsectionautorefname{Section}
\def\figureautorefname{Figure}
\def\tableautorefname{Table}
\def\assumptionautorefname{Assumption}
\def\definitionautorefname{Definition}
\def\corollaryautorefname{Corollary}
\def\algorithmautorefname{Algorithm}

\newcommand{\labels}{\ensuremath{\mathcal{A}}}
\newcommand{\graphs}{\ensuremath{\mathcal{J}}}
\newcommand{\sgraphs}{\ensuremath{\mathcal{JS}}}
\newcommand{\ren}{\text{ren}}
\newcommand{\g}{\mathcal{G}}
\newcommand{\rt}{\text{rt}}
\newcommand{\type}[1]{\textsc{#1}}
\newcommand{\emptyapp}{\textsc{App}}
\newcommand{\app}[1]{\text{\emptyapp\textsubscript{\type{#1}}}}

\title{Graphs with Multiple Sources per Vertex}
\date{March 29, 2019}
\author{Martin van Harmelen\\Supervised by Jonas Groschwitz}
\maketitle

\section{Introduction}
Instead of using logical formulas to represent meaning of natural language, recent trends use Abstract Meaning Representation (AMR) instead:
a graph based solution where nodes correspond to atomic meaning units and edges specify argument position.
Several attempts have been made at constructing them compositionally,
and recently the idea of using s-graphs with the HR-algebra \citep{koller15s-graphs}
has been simplified to reduce the number of options when parsing
\citep{groschwitz17am-algebra}%
.

This apply-modify algebra (AM-algebra)
is a linguistically plausible graph algebra with two classes of operations, both of rank two:
the apply operation is used to combine a predicate with its argument;
the modify operation is used to modify a predicate.
As terms it generates annotated s-graphs (as-graphs),
which are s-graphs annotated with a more detailed type description.

While the AM-algebra correctly handles relative clauses and complex cases of coordination,
it cannot parse reflexive sentences like: ``The raven washes herself.''
that lead to AMRs resembling the ones in \autoref{fig:amrself}.
This paper proposes a change to the type system of the AM-algebra and a change to the definition of s-graphs underlying the algebra to facilitate this.

\begin{figure}[tb]
\begin{subfigure}[b]{0.5\textwidth}
\centering
\digraph[height=4cm]{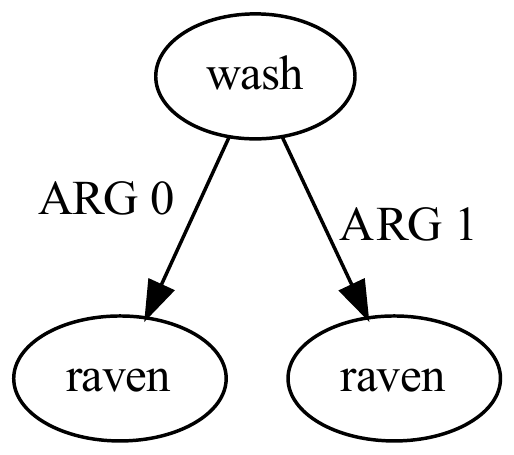}{
margin=0
    wash -> raven [xlabel="ARG 0 "]
    wash -> raven2 [label=" ARG 1"]
    raven2 [label=raven]
}
\caption{}\label{fig:amrother}
\end{subfigure}%
\begin{subfigure}[b]{0.5\textwidth}
\centering
\digraph[height=4cm]{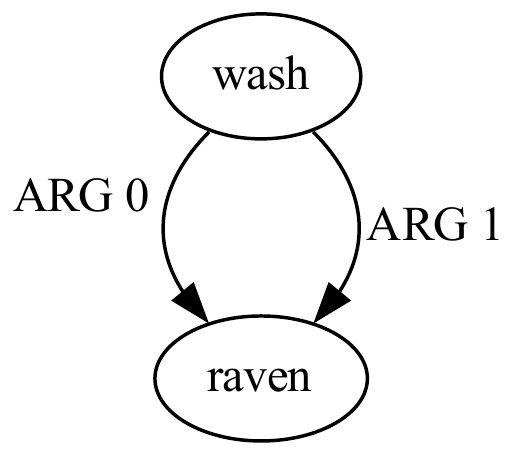}{
margin=0
    wash:sw -> raven:nw [xlabel="ARG 0 "]
    wash:se -> raven:ne [label=" ARG 1"]
}
\caption{}\label{fig:amrself}
\end{subfigure}
\caption{AMRs for ``The raven washes the raven'' (\subref{fig:amrother}) and ``The raven washes herself'' (\subref{fig:amrself})}
\end{figure}

\begin{figure}
\centering
\begin{minipage}{0.3\textwidth}
\centering
\begin{forest} sn
[ \app{s}
    [ \app{o}
        [ $G_{wash}$ ]
        [ $G_{raven}$ ]
    ]
    [ $G_{raven}$ ]
]
\end{forest}
\caption{Correct way of parsing ``The raven washes the raven''.} \label{fig:treeraven}
\end{minipage}\hfill
\begin{minipage}{0.3\textwidth}
\centering
\begin{forest} sn
[ \app{s}
    [ \app{o}
        [ $G_{wash}$ ]
        [ $G_{self}$ ]
    ]
    [ $G_{raven}$ ]
]
\end{forest}
\caption{Preferred way of parsing ``The raven washes herself''.} \label{fig:treeself}
\end{minipage}\hfill
\begin{minipage}{0.3\textwidth}
\centering
\digraph[height=2.5cm,trim=0cm 0.2cm 0cm 0.7cm]{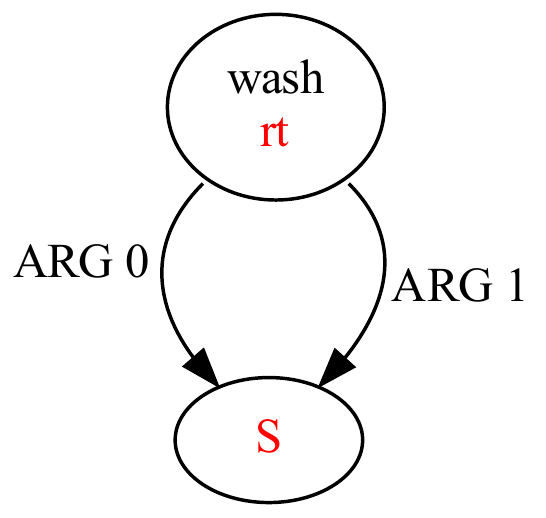}{
margin=0;
    s [label=<<FONT COLOR="red">S</FONT>>]
    wash [label=<wash<BR/><FONT COLOR="red">rt</FONT>>]
    wash:sw -> s:nw [xlabel="ARG 0 "]
    wash:se -> s:ne [label=" ARG 1"]
}
\caption{AS-graph for ``to wash oneself''.} \label{fig:amroneself}
\end{minipage}
\end{figure}

\section{The Proposal in Short} \label{sec:short}
Amongst many constructs, the AM-algebra can parse an English SVO sentence like ``The raven washes the raven'' according to the term given in \autoref{fig:treeraven} and
it should also be able to parse reflexive sentences like ``The raven washes herself'' in a similar manner, as shown in \autoref{fig:treeself}.

In this solution, the only vertex of $G_{self}$ has no label, as the necessary label is determined in the \app{s} step, but instead has an additional \type{s}-source label on its root node (\autoref{fig:selflex}).
The \app{o} operation renames \rt{} to \type{o}, composes the result with $G_{wash}$ and finally forgets the \type{o}-source label.
Because of the additional \type{s}-source label that is on the same node as the \type{o}-source label, the \type{o} and \type{s}-sources of $G_{wash}$ are merged, resulting in \autoref{fig:amroneself}.
Intuitively, the additional \type{s}-label marks that whatever position $G_{self}$ is applied to, should merge with the subject position.

The original s-graphs used in the AM-algebra, however, disallow one vertex having multiple source labels, making the lexical item $G_{self}$ needed for this derivation illegal.
Secondly, the AM-algebra demands that the type of the second argument of \app{$\alpha$} is strictly equal to the type expected by the first argument at its $\alpha$-source.
Having an extra \type{s}-label, $G_{self}$ violates this constraint.

The first part of the remainder of this paper addresses the issue regarding s-graphs, while the second part covers the type system of the AM-algebra.

\begin{figure}[tb]
\begin{subfigure}[b]{0.333\textwidth}
\centering
\digraph[height=2.5cm,trim=0cm -1.5cm 0cm -1.5cm]{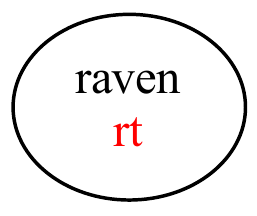}{
margin=0
    raven [label=<raven<BR/><FONT COLOR="red">rt</FONT>>]
}
\caption{$G_{raven}$}
\end{subfigure}%
\begin{subfigure}[b]{0.333\textwidth}
\centering
\digraph[height=2.5cm]{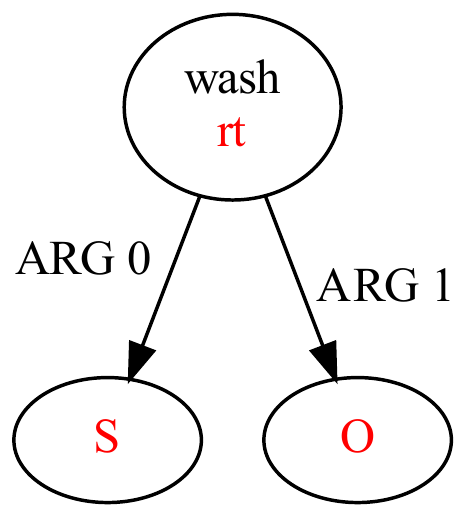}{
margin=0
    s [label=<<FONT COLOR="red">S</FONT>>]
    o [label=<<FONT COLOR="red">O</FONT>>]
    wash [label=<wash<BR/><FONT COLOR="red">rt</FONT>>]
    wash -> s [xlabel="ARG 0 "]
    wash -> o [label=" ARG 1"]
}
\caption{$G_{wash}$}
\end{subfigure}%
\begin{subfigure}[b]{0.333\textwidth}
\centering
\digraph[height=2.5cm,trim=0cm -1.5cm 0cm -1.5cm]{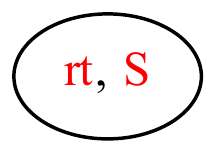}{
margin=0
    rt [label=<<FONT COLOR="red">rt</FONT>, <FONT COLOR="red">S</FONT>>]
}
\caption{$G_{self}$} \label{fig:selflex}
\end{subfigure}%
\caption{Example lexicon} \label{fig:lexicon}
\end{figure}

\section{S-graphs}
\subsection{Original Definition of S-graphs} \label{sec:original_definition}
This section reiterates the definitions by \citet{courcelle12graphs} of s-graphs and parallel-composition and concludes with a small remark.
As opposed to the more precise definition of s-graphs and parallel-composition that \citet{courcelle12graphs} give in Section~2.3, this paper works with the simpler one from Section~1.4.2:

\begin{quote}
We consider (abstract) directed or undirected graphs, possibly with multiple edges.
They form the set $\graphs$.
For a graph in $\graphs$, $E_G$ denotes its set of edges (and $V_G$ its set of vertices).
We let $\labels$ be a countable set of labels (\dots) that will be used to distinguish particular vertices.
These distinguished vertices will be called \emph{sources}, and $\labels$ is the set of \emph{source labels}.

(This notion of source is unrelated with edge directions.)

A \emph{graph with sources}, or \emph{s-graph} in short, is a pair $G = \langle G^\circ, src_G \rangle$ where $G^\circ \in \graphs$ and $src_G$ is a bijection from a finite subset $\tau(G)$ of $\labels$ to a subset of $V_{G^\circ}$.
We call $\tau(G)$ the \emph{type} of $G$ and $src_G(\tau(G))$ the set of its sources.
The vertex $src_G(a)$ is called the \emph{a-source} of $G$; its \emph{source label}, also called its \emph{source name}, is $a$.

We let $\sgraphs$ denote the set of s-graphs; (\dots) We define operations on $\sgraphs$: first a binary operation called the \emph{parallel-composition}, (\dots).
For $G, H \in \sgraphs$ we let
\[
    G \comp H := \langle G^\circ \cup H'^\circ, src_G \cup src_{H'}\rangle
\]
where $H'$ is isomorphic to $H$ and is such that
\begin{align*}
    &E_{H'} \cap E_G = \emptyset \\
    &src_{H'}(a) = src_G(a) \text{ if } a \in \tau(G) \cap \tau(H'), \\
    &V_{H'} \cup V_G = \{src_G(a) \given a \in \tau(G) \cap \tau(H') \}.
\end{align*}
This operation ``glues'' $G$ and a disjoint copy of $H$ by fusing their sources having the same names.
\end{quote}

Note that \citet{courcelle12graphs} define s-graphs with $src_G$ a bijection.
This is somewhat misleading, as ``a bijection (\dots) to a subset of $V_{G^\circ}$'' is exactly the same as saying that $src_G$ is injective to $V_{G^\circ}$.
The subset to which $src_G$ is then bijective is trivially $\img(src_G)$.

\subsection{New Graphs with Sources} \label{sec:new_definition}
To facilitate the preferred parsing method described in \autoref{sec:short},
this section gives a definition that is very similar to the one of s-graphs by \citet{courcelle12graphs}, but allows for one vertex to have multiple source labels.
Secondly, a new definition of parallel-composition is given that is equivalent to the old one when composing s-graphs, but also correctly handles the composition of graphs with sources that have more than one label.
Finally, a proof is given for this equivalence of definitions.

\subsubsection{Defining MS-graphs}
\begin{definition}[Graphs with possibly multiple source labels per vertex]
Let $\labels$ be a fixed countable set of \emph{names} or \emph{labels}.
Let $\tau(G) \subseteq \labels$, the \emph{type} of $G$, be a finite subset of $\labels$, denoting the labels used in $G$.
Instead of $src_G$ being an \emph{injective} function (as in the original definition), let $src_G: \tau(G) \to V_{G^\circ}$ be any function.
A \emph{graph with possibly multiple source labels per vertex} (\emph{ms-graph}) is a pair $G = \langle G^\circ, src_G \rangle$.
\end{definition}

This definition allows for one vertex to have multiple labels, but crucially a label still uniquely specifies a vertex.
Moreover, if $src_G$ happens to be injective, this definition of ms-graphs reduces to the definition of s-graphs, thus all s-graphs are ms-graphs.

Finally, let $slab_G\colon \img(src_G) \to \mathcal{P}\left(\tau(G)\right)$ be the inverse of $src_G$.
If $S$ is a set, we write $Slab_G(S) := \bigcup \{ slab_G(s) \given s \in S \cap \dom(slab_G) \}$.

\subsubsection{Redefining parallel-composition}
This section shows the definition by \citet{courcelle12graphs} of parallel-composition does not work for ms-graphs and gives a new definition.

The example in \autoref{fig:compprob} together with the following corollary shows the original definition of parallel-composition is contradictory when used on ms-graphs.

\begin{corollary}\label{cor:disjoint}
Let $G$ and $H$ be s-graphs, $\abs{\tau(G)} \geq 1$ and, without loss of generality, let $a, b \in \tau(G)$.
If $src_G(a) \neq src_G(b)$%
, then $src_{G \comp H}(a) \neq src_{G \comp H}(b)$,
because $\restr{src_{G \comp H}}{\tau(G)} = src_G$.
\end{corollary}

Crucially, \autoref{cor:disjoint} states that all sources that were distinct within $G$ are still distinct in $G \comp H$.
This poses a problem for composing the graphs in \autoref{fig:compprob}, as $G_{separate} \comp G_{single}$ should be isomorphic to $G_{single}$.
By \autoref{cor:disjoint}, however, it must have at least as many nodes as $G_{separate}$, which is a contradiction.

\begin{figure}[tb]
\begin{subfigure}[b]{0.5\textwidth}
\centering
\digraph[width=\textwidth]{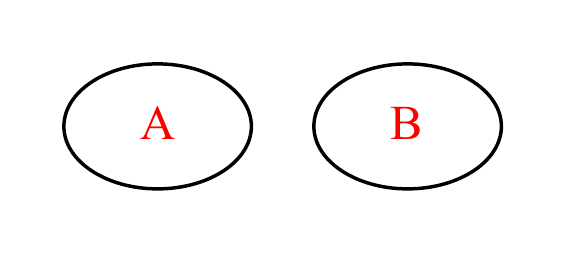}{
margin=0.2
    a [label=<<FONT COLOR="red">A</FONT>>]
    b [label=<<FONT COLOR="red">B</FONT>>]
}
\caption{$G_{separate}$}
\end{subfigure}%
\begin{subfigure}[b]{0.5\textwidth}
\centering
\digraph[width=\textwidth]{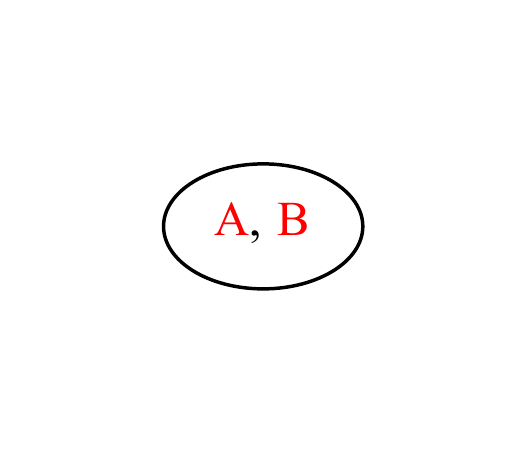}{
margin=0.6
    rt [label=<<FONT COLOR="red">A</FONT>, <FONT COLOR="red">B</FONT>>]
}
\caption{$G_{single}$}
\end{subfigure}%
\caption{Two graphs that would pose a problem when composed using the original definition of parallel-composition.} \label{fig:compprob}
\end{figure}

For the more robust definition of parallel-composition of $G$ and $H$, an equivalence relation on the disjoint union of the vertices of $G$ and $H$ is needed that makes sources with overlapping labels equivalent.

\begin{definition} \label{def:eq}
Let $G$ and $H$ be any two ms-graphs and let $H'$ be isomorphic to $H$ such that $H'$ shares no vertices or edges with $G$.
Let $V_\sim := V_G \cup V_{H'}$ and let $R$ be a binary relation on $V_\sim$ such that $\forall g \in V_G, h \in V_{H'}$:
\begin{equation*}
    (g, h) \in R \iff \exists \alpha \in \labels \colon g = src_G(\alpha) \land h = src_{H'}(\alpha)
\end{equation*}
Let $\sim_{G \comp H}$ (or $\sim$ if $G$ and $H$ are understood from context) be the symmetric and transitive closure of $R$ and additionally, for all $v \in V_\sim$: let $(v, v) \in \ \sim$.
\end{definition}

Note that $\sim$ is an equivalence relation on $V_\sim$.

\begin{definition}[Parallel-composition of ms-graphs] \label{def:comp}
Let $G$, $H$, $H'$, $V_\sim$ and $\sim$ be as in \autoref{def:eq}.
Let $X \subseteq V_\sim$ be a cross-section of $\sim$ as in the definition of quotient graphs by \citet[Definition~2.15]{courcelle12graphs}, which will also be used here.
Furthermore, let:
\begin{align*}
    \tau(G \comp H) &:= \tau(G) \cup \tau(H)\footnotemark \\
    [x] &:= \{v \in V_\sim \given v \sim x \} & \forall x \in X
\end{align*}
\footnotetext{Note that isomorphic ms-graphs have the same type.}
$slab_{G \comp H}$:
\begin{align*}
        X &\to \mathcal{P}\left(\tau(G \comp H)\right) \\
        x &\mapsto Slab_G([x]) \cup Slab_{H'}([x])
\end{align*}
$src_{G \comp H}$:
\begin{align*}
        \tau(G \comp H) &\to X \\
        \alpha &\mapsto \begin{cases}
        x &\text{if }
            \alpha \in slab_{G \comp H}(x)  \\
        \text{undefined} 
    \end{cases}
\end{align*}
Then the parallel-composition of $G$ and $H$ is:
\begin{equation*}
    G \comp H := \langle \left(\left( G^\circ \cup H'^\circ \right) / \sim\right)_X, src_{G \comp H} \rangle
\end{equation*}
\end{definition}

The following theorem suffices to show this definition of parallel-composition is well-defined.
\begin{theorem}
$src_{G \comp H}$ is well-defined.
\end{theorem}
\begin{proof}
Let $x, y \in X$ and $\alpha \in \tau(G \comp H)$.
If $\alpha \in slab_{G \comp H}(x)$ and $\alpha \in slab_{G \comp H}(y)$,
then there exist $u, v \in V_\sim$ such that $u \sim x$, $v \sim y$ and one of the following two cases holds:
\begin{enumerate}
    \item $slab_G(u) \ni \alpha \in slab_G(v)$ or $slab_{H'}(u) \ni \alpha \in slab_{H'}(v)$
    \item $slab_G(u) \ni \alpha \in slab_{H'}(v)$ or $slab_{H'}(u) \ni \alpha \in slab_{G}(v)$
\end{enumerate}

The first case implies that either $u = src_G(\alpha) = v$ or $u = src_{H'}(\alpha) = v$, thus by transitivity of $\sim$ we have $x \sim y$.

Without loss of generality, let us assume $u \in G$ and $v \in H'$ in the second case.
Then $u = src_G(\alpha)$ and $v = src_{H'}(\alpha)$, thus by \autoref{def:eq} $u \sim v$ and by transitivity of $\sim$ we have $x \sim y$.

$X$, however, is a cross-section of $\sim$, thus $x \sim y$ implies $x = y$.
\end{proof}

Note that, like the definition of parallel-composition by \citet{courcelle12graphs}, \autoref{def:comp} only determines the parallel-composition of two ms-graphs up to isomorphism.

\subsection{Proof of Equivalence of Definitions}
The following shows that \autoref{def:comp} reduces to the original definition of parallel-composition when composing s-graphs.
\begin{proof}
Let $G$, $H$, $H'$, $V_\sim$ and $\sim$ be as in \autoref{def:eq} and specifically, let $G$ and $H$ be s-graphs.
In this proof $A \cong B$ denotes that $A$ is isomorphic to $B$ and $G \comp H$ refers to parallel-composition as in \autoref{def:comp}.
Finally, let $H''$ be what \citeauthor{courcelle12graphs} denote by $H'$.

What must be shown is that
\begin{equation*}
    \left(\left( G^\circ \cup H'^\circ \right) / \sim\right)_X \cong G^\circ \cup H''^\circ
\end{equation*}
with some isomorphism $c$, such that
\begin{equation*}
    c \circ src_{G \comp H} = src_G \cup src_{H''}.
\end{equation*}

By definition of s-graphs,
every vertex of $G$ and $H'$ has at most one source
and thus for every source in $G$ there is at most one equivalent vertex,
which is the vertex of $H'$ with the same source label, if it exists.
Thus $\forall g, g' \in V_{G}$:
\begin{equation} \label{eq:G}
    g \sim g' \iff g = g',
\end{equation}
similarly $\forall h, h' \in V_{H'}$:
\begin{equation} \label{eq:H}
    h \sim h' \iff h = h'
\end{equation}
and finally $\forall g \in V_G, h \in V_{H'}$:
\begin{equation} \label{eq:GH}
    g \sim h \iff \exists \alpha \in \labels \colon g = src_G(\alpha) \land h = src_{H'}(\alpha).
\end{equation}

Let $K^\circ$ be a subgraph of $\left(\left( G^\circ \cup H'^\circ \right) / \sim\right)_X$ (or $(G \comp H)^\circ$) such that
\begin{align*}
    V_K &:= \{x \in X \given \exists h \in V_{H'} \colon h \sim x\} \\
    E_K &:= E_{H'} \\
    vert_K(e) &:= (x, y) &\forall e \in E_K
\end{align*}
with $x, y \in X$, $x \sim h$, $y \sim h'$ and $vert_{H'}(e) = (h, h')$\footnote{Such $x$ and $y$ exist and are unique, because $X$ is a cross-section of $V_\sim$.}.

In other words, $K^\circ$ is the part of $(G \comp H)^\circ$ that came from $H'$.
Let $src_K := \restr{src_{G \comp H}}{\tau(H')}$ and let $K := \langle K^\circ, src_K \rangle$.
\autoref{eq:G} and \ref{eq:H} and \ref{eq:GH} together imply $\abs{V_K} = \abs{V_{H'}}$ and thus by construction
$K \cong H'$.

The correct choice of $X$ makes $K$
satisfy all requirements of $H''$ and thus $H''$ can be chosen equal to it, making $c$ equal to the identity,
proving the equivalence of definitions.
The remainder of this proof describes this choice of $X$, verifies that $K$ then satisfies the requirements for $H''$ and finally formally checks equality.

By \autoref{eq:G}, we can choose $X$ such that $V_G \subseteq X$
and therefore
$slab_{G \comp H}(g)$ is defined for all $g \in V_G$.
By \autoref{eq:GH}%
, we have $\forall g \in V_G$:
\begin{equation*}
    Slab_G([g]) \supseteq Slab_{H'}([g])
\end{equation*}
and thus
\begin{align*}
    slab_{G \comp H}(g) &= Slab_G([g]) \cup Slab_{H'}([g])
    & \forall g \in V_G \\
    &= Slab_G([g]) \\
    &= slab_G(g) & \text{by \autoref{eq:G}.}
\end{align*}
Therefore, $\forall \alpha \in \tau(G), g \in V_G$:
\begin{align*}
    \alpha \in slab_{G \comp H}(g) &\implies src_{G \comp H}(\alpha) = g
    &\text{by \autoref{def:comp}}\\
    \alpha \in slab_{G}(g) &\implies src_{G \comp H}(\alpha) = g \\
    src_G(\alpha) = g &\implies src_{G \comp H}(\alpha) = g \\
    \restr{src_{G \comp H}}{\tau(G)} &= src_G. \autotag{eq:restrG}
\end{align*}

This paragraph checks $K$ is suitable as $H''$.
By choice of $H'$, we have
\begin{equation*}
    E_{K} \cap E_G = \emptyset.
\end{equation*}
By definition of $K$ and \autoref{eq:restrG}, we have $\forall \alpha \in \tau(G) \cap \tau(K)$:
\begin{equation*}
    src_K(\alpha) = src_{G \comp H}(\alpha) = src_G(\alpha).
\end{equation*}
For the last condition,
\begin{equation*}
    V_G \cap V_K = \{ g \in V_G \given \exists h \in V_{H'} \colon h \sim g \}
\end{equation*}
by definition of $V_K$.
\autoref{eq:GH} implies $\forall g \in V_G$:
\begin{align*}
    g \in V_G \cap V_K
    &\iff
    \exists h \in V_{H'}, \alpha \in \labels \colon g = src_G(\alpha) \land src_{H'}(\alpha) = h
\end{align*}
thus
\begin{align*}
    V_G \cap V_K &= \{ src_G(\alpha) \given \alpha \in \tau(G) \land \left[\exists h \in V_{H'} \colon src_{H'}(\alpha) = h\right] \} \\
    V_G \cap V_K &= \{ src_G(\alpha) \given \alpha \in \tau(G) \cap \tau(H') \} \\
    V_G \cap V_K &= \{ src_G(\alpha) \given \alpha \in \tau(G) \cap \tau(K) \}
\end{align*}
for $\tau(K) = \tau(H')$.

Finally, let us check $(G \comp H)^\circ = G^\circ \cup K^\circ$ and $src_{G\comp H} = src_G \cup src_K$.
$K^\circ \subseteq (G \comp H)^\circ$ by choice of $K^\circ$ and similarly $G^\circ \subseteq (G \comp H)^\circ$ by choice of $X$, thus $(G \comp H)^\circ \supseteq G^\circ \cup K^\circ$.
For the inverse, we must check that $\forall x \in X$:
\begin{equation} \label{eq:supset}
    x \not \in V_G \to x \in V_k.
\end{equation}
$X$ is a cross-section of $\sim$, thus it is a subset of $V_{\sim} = V_G \cup V_{H'}$, thus $\forall x \in X \colon x \in V_G \lor x \in V_{H'}$.
If $x \in V_G$, then \autoref{eq:supset} holds.
If $x \not \in V_G$, then $x \in V_{H'}$ and therefore $\exists h \in V_{H'} \colon h \sim x$, namely $h = x$ and thus $x \in V_{K}$,
proving that $X = V_G \cup V_K$.

For the edges:
\begin{align*}
    E_{(G \comp H)^\circ} &= E_{G^\circ \cup H'^\circ}
        & \text{by definition of quotient graphs} \\
        &= E_{G} \cup E_{H'} \\
        &= E_{G} \cup E_{K}
        & \text{by choice of } E_K.
\end{align*}

Verifying $vert_{(G \comp H)^\circ} = vert_G \cup vert_{K}$ it too tedious for the current proof,
but boils down to showing that for all $e \in E_{H'}$ either $vert_{H'}(e) = vert_K(e)$ or one or both of the vertex instances of $vert_{H'}(e)$ are not in $X$, but then they are replaced with their unique equivalent $x \in X$, exactly as in the definition of $vert_K$.

This shows that also $(G \comp H)^\circ \subseteq G^\circ \cup K^\circ$ and thus $(G \comp H)^\circ = G^\circ \cup K^\circ$.
Because of this equality and because $\restr{src_{G \comp H}}{\tau(K)} = src_K$ by definition and $\restr{src_{G \comp H}}{\tau(G)} = src_G$ (\autoref{eq:restrG}), we also have $src_{G \comp H} = src_G \cup src_K$.
\end{proof}

\section{Type-System of the AM-algebra}

A graph type of an as-graph $g$ annotates each source label $\alpha$ with the type of the graph that can be used as argument to $\app{$\alpha$}\left(g, - \right)$ \citep[Definition~3.1]{groschwitz17am-algebra}.
Note that using ms-graphs allows multiple and thus differing type restrictions
when apply is used on a source node,
but because the desired source name must be specified during application,
this does not lead to contradictions.

What \emph{is} problematic in combining $G_{wash}$ with $G_{self}$ as shown in \autoref{fig:treeself},
is that the graph type of $G_{wash}$ expects $G_{self}$ to have empty type, but $G_{self}$ is of type \type{s},
rendering $\app{s}\left(G_{wash}, G_{self}\right)$ undefined,
according to Definition~3.3 of \citet{groschwitz17am-algebra}, cited below.
This section first recounts said Definition~3.3, lists the problems that arise when relaxing it and finally proposes changes that take these problems into account.

\subsection{Original Definition of Apply Operation}
\begin{quote}
\begin{quotedef}{3.3}[Apply operation (\emptyapp)]
\ \\
Let $\g_1 = \left(\left(g_1, S_1\right), \left(T_1, R_1\right)\right)$, $\g_2 = \left(\left(g_2, S_2\right), \left(T_2, R_2\right)\right)$ be as-graphs.
Then we let $\app{$\alpha$}\left(\g_1, \g_2\right) = \left(\left(g', S'\right), \left(T', R'\right)\right)$ such that
\begin{align*}
    (g', S') &= f_\alpha( (g_1, S_1) \comp \ren_{\{\rt \mapsto \alpha \}}(\ren_{R(\alpha)}( (g_2, S_2)))) \\
    T' &= (T_1 \setminus \{\alpha\}) \cup (T_2 \circ \overline{R_1 (\alpha)^{-1}})\\
    R' &= (R_1 \setminus \{\alpha\}) \cup (R_2 \circ \overline{R_1 (\alpha)^{-1}})
\end{align*}
if and only if
\begin{enumerate}
    \item $\g_1$ actually has an $\alpha$-source to fill, i.e. $\alpha \in \dom(T_1)$
    \item \label{cond:original} $\g_2$ has the type $\alpha$ is looking for, i.e. $T_1(\alpha) = (T_2, R_2)$, and
    \item $T', R'$ are well-defined (partial) functions;
\end{enumerate}
otherwise $\app{$\alpha$}(\g_1, \g_2)$ is undefined.
\end{quotedef}
\end{quote}

\subsection{Problems}
The only place where this definition has to be relaxed, is in Condition~\ref{cond:original}.
To keep the type system functional, this relaxation should not be too big:
the type of the output graph using the original definition must not be violated.

Simply allowing $\g_2$ to have any extra annotation that $\g_1$ already contains (conform the modify operation \citep[Definition~3.4]{groschwitz17am-algebra}) allows the term in \autoref{fig:treeself} and does not change the output type.
At first this seems fine, but there are three problems:
\begin{enumerate}
    \item \label{prob:wash} This would also allow \type{s}-application of $G_{wash}$ to itself, which is definitely not desirable,
    as intuitively transitive verbs should only be allowed to combine with entities and not with other verbs.
    \item \label{prob:self} Moreover, this would allow \type{s}-application of $G_{self}$ to something, which is undesirable because
    the \type{s}-label in $G_{self}$ is there only to signify a merge with the subject of the sentence-to-be, not to label a to-be-filled argument slot.
    \item \label{prob:app} Finally, this would allow something to be \type{s}-applied to $G_{self}$, making the extra \type{s}-label redundant.
    This \emph{could} be undesirable, as the spirit of the AM-algebra is to reduce the number of terms leading to the same AMR\footnote{
    This all hinges around whether $ren_{\{\text{rt} \mapsto \alpha \}}(g)$ is defined if $\alpha \in \tau(g)$.
    For if it is not defined,
    then the AM-algebra implicitly disallows $\alpha$-application of something to a graph with an $\alpha$-source that is not renamed,
    which indicates Problem~\ref{prob:app} is indeed undesirable,
    but also implicitly disallowed.
    \citeauthor{courcelle12graphs}, however, only define a simultaneous rename operation where labels \emph{swap} position,
    which would lead to weird but not explicitly prevented results.
    \citeauthor{groschwitz17am-algebra} do not bother specifying whether their rename works the same.
    }.
\end{enumerate}

All three problems signify that any additional labels of the root node should be treated differently than the source labels of other nodes.
This does not only involve relaxing Condition~\ref{cond:original}, but also requires adding conditions.

\subsection{Proposed Changes}
For brevity, let $rlab(G) := slab_{G}\left(src_{G}(\rt)\right) \setminus \{\rt\}$ for any ms-graph, the set of \emph{additional root labels} of $G$.
Problem~\ref{prob:self}
is
simply solved by adding the following condition:
\begin{enumerate}
    \setcounter{enumi}{3}
    \item \label{cond:notaddroot} $\alpha \notin rlab(\g_1)$, i.e. $\alpha$ is not an additional root label of $\g_1$,
\end{enumerate}
or instead, if Problem~\ref{prob:app} must be accounted for explicitly:
\begin{enumerate}
    \setcounter{enumi}{3}
    \item $\alpha \notin rlab(\g_1) \cup rlab(\g_2)$, i.e. $\alpha$ is an additional root label of neither $\g_1$ nor $\g_2$.
\end{enumerate}

Problem~\ref{prob:wash} can be solved by making sure not to relax Condition~\ref{cond:original} too much.
A safe choice would therefore be to relax Condition~\ref{cond:original} as little as possible, while still reaching our goal.

The minimum relaxation needed for the term in \autoref{fig:treeself} to be legal, would be to allow at most one additional \type{s}-label at the root.
This, however, would be tailored too much to the specific choice of label and number of additional labels.
Instead, the minimum relaxation that does not specify the (number of) additional root labels, is the replacement of Condition~\ref{cond:original} by the following conditions:
\begin{enumerate}[2a]
    \item \label{cond:ignoretype} $T_1(\alpha) = \left(T_2 \setminus rlab(g_2), R_2 \setminus rlab(g_2)\right)$, i.e. apart from its additional root labels, $\g_2$ has the type $\alpha$ is looking for,
    \item \label{cond:subtype} $\restr{T_2}{rlab(g_2)} \subseteq T_1$ and $\restr{R_2}{rlab(g_2)} \subseteq R_1$, i.e. the additional root labels of $\g_2$ do not change the type of $\g_1$.
\end{enumerate}

This changed \emptyapp{}  operation is trivially equivalent to the original one when only using s-graphs:
for any s-graph $g$, $rlab(g) = \emptyset$, by definition of s-graphs.
Thus Condition~\ref{cond:notaddroot} and Condition~\ref{cond:subtype} become tautologies%
,
and Condition~\ref{cond:ignoretype} reduces to Condition~\ref{cond:original} of the original definition.

\section{Conclusion}
The AM-algebra by \citet{groschwitz17am-algebra} is powerful and linguistically plausible,
but not powerful enough to parse reflexive sentences in a linguistically preferred way.
In summary, this paper proposed a change to the AM-algebra and the s-graphs underlying it,
and showed the proposed definitions reduce to the original definitions when used under the original constraints.
Most importantly, these changes enable the AM-algebra to parse reflexive sentences in a linguistically preferred way.

\bibliographystyle{plainnat}
\bibliography{bibliography}

\end{document}